\documentclass[aps,preprint]{revtex4}%
\usepackage{amsfonts}
\usepackage{amsmath}
\usepackage{amssymb}
\usepackage{graphicx}%
\newtheorem{theorem}{Theorem}

\newtheorem{lemma}{Lemma}

\newenvironment{proof}[1][Proof]{\noindent\textbf{#1.} }{\ \rule{0.5em}{0.5em}}
\begin{document}
\preprint{Internal Communication }
\title{The cmb dipole and existence of a center for expansion of the universe}
\author{Yukio Tomozawa}
\affiliation{Michigan Center for Theoretical Physics}
\affiliation{Randall Laboratory of Physics, University of Michigan}
\affiliation{Ann Arbor, MI. 48109-1040, USA}
\date{\today }

\begin{abstract}
It is shown that the observed cmb dipole implies the existence of a center for
the expansion of the universe, and that it can be explained by a combination
of peculiar velocity and Hubble flow.

\end{abstract}

\maketitle

\section{\label{sec:level1}Introduction}

In the Friedman universe, one possible interpretation of the coordinates is
that the whole space is on the surface of an expanding balloon and has no
center. (A center exists outside the universe in a way.) This accommodates the
Hubble law naturally. There are no velocities associated with points of the
universe, but the relative distance and relative velocity of any two points
increase with the expansion of the balloon. I will show that in such a
universe, there is no cosmic microwave backgroud (cmb) dipole even in the
presence of a peculiar velocity. In other words, the observation of a cmb
dipole excludes such an interpretation of the coordinates for the Friedman universe.

\section{The CMB Dipole in the Temperature Distribution in Cosmology without
Center}

A dipole component was observed in the temperature distribution in the cmb
measurement. The cmb blueshift dipole for the solar system \cite{dipole} is
given by%
\begin{equation}
v=371\pm0.5\text{ }km/s,\text{ \ \ \ }l=264.4\pm0.3%
{{}^\circ}%
,\text{ \ \ \ }b=48.4\pm0.5%
{{}^\circ}%
.
\end{equation}
Or equivalently, the cmb dipole for redshift is%
\begin{equation}
v=371\pm0.5\text{ }km/s,\text{ \ \ \ }l=84.4\pm0.3%
{{}^\circ}%
,\text{ \ \ \ }b=-48.4\pm0.5%
{{}^\circ}%
. \label{dipole}%
\end{equation}
In the absence of peculiar velocity, there is no cmb dipole in a cosmology
without a center. I will state that as a lemma.

\begin{lemma}
There is no cmb dipole at any point of the universe in a cosmology without
center, in the absence of a peculiar velocity.
\end{lemma}

\begin{proof}
This is almost self-evident. In any direction from a point in the universe,
the distance $l_{0}$ from a cmb emitter to a selected point becomes $l$ after
expansion and the redshift factor is given by%
\begin{equation}
1+z=\frac{l}{l_{0}} \label{eq1}%
\end{equation}
and this value is the same for all directions. Of course, differences in the
redshift or the temperature distribution in the cmb measurement come from the
structure variation of the emitters, which is the whole issue of the cmb phenomenon.
\end{proof}

The following theorem is surprising.

\begin{theorem}
There is no cmb dipole at any point in the universe in a cosmology without
center, even in the presence of a peculiar velocity $v_{p}$.
\end{theorem}

I will present the proof in 3 ways.

\begin{proof}
I. Seen from the rest frame of a peculiar velocity, both $l_{0}$ and $l$ are
Lorentz contracted by the same factor $\sqrt{1-(v_{p}\cos\theta/c)^{2}}$,
where $\theta$ is the angle between the emitter and the peculiar velocity, and
their ratio in Eq. (\ref{eq1}) is unchanged. This is true for all directions.
\end{proof}

\begin{proof}
II. Relating the equivalent velocity of the cmb emitter $v$ to the expansion
rate $1+z$ by%
\begin{equation}
\sqrt{\frac{1+v/c}{1-v/c}}=1+z,
\end{equation}
one gets%
\begin{equation}
v/c=\frac{(1+z)^{2}-1}{(1+z)^{2}+1}=1-2\frac{1}{(1+z)^{2}}=1-2\times10^{-6}
\label{eq2}%
\end{equation}
for $z=1000$. The relative velocity of the emitter and the peculiar velocity
$v_{p}$ in the direction of the emitter is%
\begin{equation}
\frac{v-v_{p}\cos\theta}{1-vv_{p}\cos\theta/c^{2}}=v-v_{p}\cos\theta
+(v/c)^{2}v_{p}\cos\theta=v-O(4\times10^{-6}v_{p}\cos\theta)
\end{equation}
It is easy to see that this result is valid in any direction.
\end{proof}

\begin{proof}
III. An object that moves with peculiar velocity $v_{p\text{ }}$ is at rest
with respect to an object at a distance of $v_{p}/H_{0}$, where $H_{0}$ is the
Hubble constant, which does not have a cmb dipole by the Lemma. Therefore, an
object with a peculiar velocity should not have a cmb dipole.
\end{proof}

All three proofs give the same result. Another way to look at this theorem is
that the equivalent speed of a cmb emitter is close to that of light and speed
of light is identical for moving frames. We have reached the important
conclusion that in a cosmology without center there is no cmb dipole. As a
corollary, we state a theorem.

\begin{theorem}
The observation of the cmb dipole excludes the possibility of a cosmology
without center. Thus, there has to be a center for the expansion of the
universe, since a cmb dipole has been observed for the solar
system\cite{dipole}.
\end{theorem}

The author has discussed the implications of the cmb dipole for cosmologies
with centers elsewhere\cite{center}, to which we refer for details. There, the
observed cmb dipole is related to a combination of Hubble flow and peculiar velocity.

\begin{acknowledgments}
The author would like to thank Jean P. Krisch and David N. Williams for useful information.
\end{acknowledgments}

Correspondence should be addressed to the author at tomozawa@umich.edu.

\end{document}